\RequirePackage{fix-cm}

\documentclass[12pt]{article}
\usepackage[margin=0.9in]{geometry}

\usepackage{graphicx}
\usepackage{fontenc,algorithm,algorithmic}
\usepackage{amsmath}
\usepackage{amsthm}  
\usepackage{multirow}

\begin{document}
\title{A Near Optimal Approximation Algorithm for Vertex-Cover Problem}
\author{Deepak Puthal\footnote{deepakpnitr@gmail.com} \\
Department of Computer Science \& Engineering\\
National Institute of Technology Silchar\\
Silchar, Assam, India}

\maketitle
\begin{abstract}
Recently, there has been increasing interest and progress in improvising the approximation algorithm for well-known NP-Complete problems, particularly the approximation algorithm for the Vertex-Cover problem. Here we have proposed a polynomial time efficient algorithm for vertex-cover problem for more approximate to the optimal solution, which lead to the worst time complexity  $\Theta(V^2)$ and space complexity $\Theta(V+E)$. We show that our proposed method is more approximate with example and theorem proof. Our algorithm also induces improvement on previous algorithms for the independent set problem on graphs of small and high degree.
\end{abstract}
{\bf Keywords}:Approximation algorithm \and Vertex-Cover Problem \and Complexity \and Adjacency list.

\section{Introduction}
\label{intro}
A graph G represents as G = (V, E): V is number of vertices in graph and E is number of edges in graph and impliment as an adjacency lists or as an adjacency matrix for both directed and undirected graphs. There are two types of graph i.e. (I) Sparse graphs-those for which $|E|$ is much less than $|V|^2$ $(E<<V^2)$. (II) Dense graphs-those for which $|E|$ is close to $|V|^2 (E \simeq  V^2)$. Here we presented the graph as the adjacency-list for the evaluation of our algorithm. \\
The vertex-cover problem is to find a minimum number of vertex to cover a given undirected graph. We call such a vertex cover an optimal vertex cover. This problem is the optimization version of an NP-complete decision problem. Proposed algorithm is polynomial time algorithm in order to find the set of vertex to cover the graph. Which shows the better performance than the traditional algorithm for vertex cover  \cite{RefB1}. 

\section{The vertex-cover problem}
As it is NP-Hard problem so it is hard to find an optimal solution of a graph G, but not difficult to find a near optimal solution. Our propose method gives very near optimal solution for  Vertex-cover problem. The following approximation algorithm takes an undirected graph G as input  \cite{RefB1} and returns a set of vertex to cover the graph and whose size is less than the previous method. \\
All graphs mentioned here are simple undirected graph. We follow \cite{RefB2} for definitions. Our proposed method (See algorithm and Fig.\ref{fig:2}) is on undirected graph. Here we used the adjacency list to represent graph G. We introduce a new field \textit{weight} in the \textit{list} to store the degree of each individual vertex. \textit{i.e.}  

struct list

\{

\vspace{1 mm}    \hspace{5 mm}    char vertex;

\vspace{1 mm}    \hspace{5 mm}    int weight;

\vspace{1 mm}    \hspace{5 mm}   struct node *next;

\vspace{1 mm}    \hspace{5 mm}    struct node *ref;

\};

\begin{algorithm}
\caption{Approximate Vertex-Cover Algorithm}
\begin{algorithmic} [1]
\REQUIRE {{In the \textit{List} we introduce another field \textit{weight} \\    The value of \textit{weight} is number of node in reference \textit{(ref)}}}
\STATE $C^+ \leftarrow  \emptyset $
\STATE L = \textit{List}
\STATE L[w] = Reference weight 
\STATE (h, v) = highest weight of the list and respective vertex  \label{marker}
\IF {$h \neq 0$} 
	\STATE $C^+  \leftarrow C^+ \cup {v}$
	\STATE v[w] $\leftarrow 0$
	\FOR{all vertex of \textit{List} L[ref] $\in \{v\}$}
		\STATE L[w] $\leftarrow$ L[w]-1
	\ENDFOR
	\STATE \textbf{go to} \ref{marker}
\ELSE
	\RETURN $C^+$
\ENDIF
\end{algorithmic}
\end{algorithm}

It's space complexity is $\Theta(V+E)$ \cite{RefB1}.  For step 4 search in the list is O(V). In the for loop \textit{i.e. step 8 to 10} for each individual vertex need to search its reference vertices. So worst time complexity is $O(V*(V-1))$= $O(V^2)$. So the worst time complexity of the graph is $\Theta(V^2)$.

\newtheorem{mydef}{Theorem}
\begin{mydef}
\emph{(Thomas H. Cormen et.al. \cite{RefB1})}
\label{vcp1}
APPROX-VERTEX-COVER is a polynomial-time 2-approximation algorithm.
\end{mydef}

\begin{mydef}
\label{vcp2}
Proposed Approximate Vertex-Cover is a polynomial-time $(2-\varepsilon)$-approximation algorithm.
\end{mydef}

\begin{proof}
In Theorem ~\ref{vcp1} $C$ is the set of vertex for APPROX-VERTEX-COVER and $C^*$ in the optimal vertex cover $i.e. |C| \leq 2|C^*|$. In our approach we pick one vertex and remove the edges connected to that vertex. So most of the times we don't consider both end point of one edges, which followed in Theorem ~\ref{vcp1}. For our proposed method we consider the resultant set of vertex is $C^+$, then $|C^+| = |C| - \varepsilon$.\\
$ \Rightarrow |C^+| = (2- \varepsilon) |C^*|$, $0  \leq  \varepsilon   \leq  1$.
\qed
\end{proof}

In some cases proposed method$(C^+)$ approaches to optimal solution when $\varepsilon$ value is 1. Our method is shown in Fig. \ref{fig:1}, algorithm (Approximate Vertex-Cover), and proved in Theorem ~\ref{vcp2}. The comparison of the optimal vertex-cover, previous vertex-cover and proposed vertex-cover shown in Fig. \ref{fig:2}.

\begin{figure}
  \includegraphics[width=0.65\textwidth]{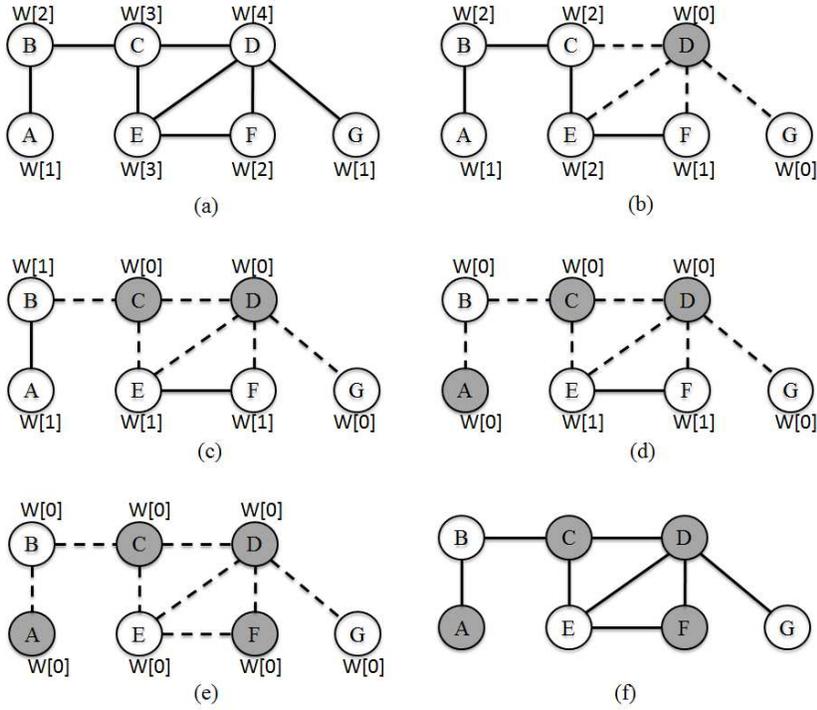}
\caption{The operation of Approximate Vertex-Cover. \textit{(a)} The input graph G, which has 7 vertices and 8 edges. \textit{(b)} The vertex \textit{D} has the highest weight, is the first vertex chosen by Approximate Vertex-Cover. The weight of the vertex initiate to zero and the vertices directly connected from vertex \textit{D} weight value decrees by 1. The vertex \textit{D} is shaded and the edges is dashed incident from vertex \textit{D}. The vertex \textit{D} add to the set $C^+$. \textit{(c)} There are three vertex with highest weight (\textit{B, C, E}); Arbitrary vertex \textit{C} is chosen and add to the set $C^+$. \textit{(d)} Vertex \textit{A} is chosen; added to $C^+$. \textit{(e)} Vertex \textit{F} is chosen; added to $C^+$. \textit{(f)} The resultant vertex cover with our proposed method. Approximate Vertex-Cover, contains the four vertices \textit{a, c, d, f}. }
\label{fig:1}       
\end{figure}
%
\begin{figure*}
  \includegraphics[width=0.65\textwidth]{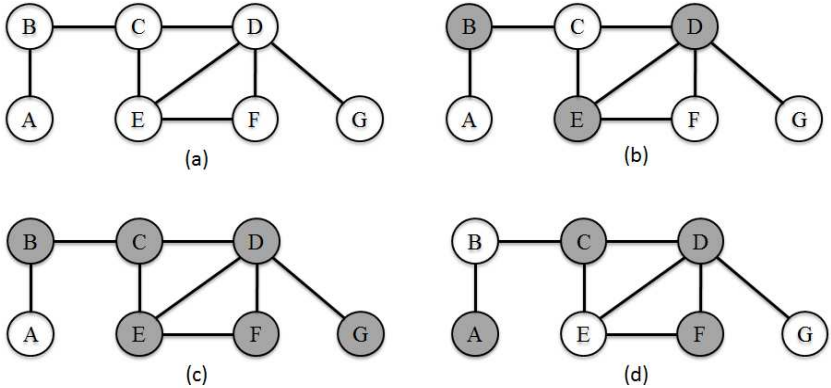}
\caption{Comparison of optimal, Traditional and proposed vertex cover result. \textit{(a)} The inputted original graph G. \textit{(b)} The optimal result of input graph. \textit{(c)} The result with traditional method. \textit{(d)} The proposed algorithm's result for vertex-cover algorithm.}
\label{fig:2}       
\end{figure*}

\section{Conclusion}
Here in our proposed technique we produce the set of vertex for vertex-cover problem. Which is more near optimal solution and better than the previous technique.

\end{document}